\documentclass[journal,12pt,onecolumn]{IEEEtran}

\usepackage{amsthm}
\usepackage{amsmath,amssymb,mathtools}
\usepackage{cite}
\usepackage[colorlinks=true,linkcolor=blue,citecolor=blue,urlcolor=blue]{hyperref}

\newtheorem{theorem}{Theorem}
\newtheorem{lemma}[theorem]{Lemma}
\newtheorem{proposition}[theorem]{Proposition}

\theoremstyle{definition}
\newtheorem{definition}[theorem]{Definition}
\theoremstyle{plain}

\newcommand{\F}{\mathbb F}
\newcommand{\ind}{\mathbf{1}}
\newcommand{\wt}{\operatorname{wt}}
\newcommand{\Prob}{\mathbb{P}}
\newcommand{\Exp}{\mathbb{E}}
\newcommand{\GL}{\operatorname{GL}}
\DeclareMathOperator{\rank}{rank}
\DeclareMathOperator{\Row}{Row}

\begin{document}

\title{The Optimal Asymptotic Rate of Generalized Covering Codes}
\author{Hengzhuo~Li, Chong~Shangguan, and Hengjia~Wei%
\thanks{H. Li and H. Wei are with the School of
Mathematics and Statistics, Xi'an Jiaotong University, Xi'an 710049,
China.  Emails: leeker0626@outlook.com (H. Li) and hjwei05@gmail.com (H. Wei).
The work of H. Li and H. Wei was supported by the National Natural Science
Foundation of China under Grant 12371523.}%
\thanks{C. Shangguan is with the Research Center for
Mathematics and Interdisciplinary Sciences, Shandong University, Qingdao
266237, China, and the Frontiers Science Center for Nonlinear Expectations,
Ministry of Education, Qingdao 266237, China.
Email: theoreming@163.com. C. Shangguan was supported by the National Natural Science Foundation of
China under Grant Nos.~12571352 and~12231014, and the Fundamental Research
Funds for the Central Universities.}%
\thanks{H. Wei is the corresponding author.}}

\maketitle

\begin{abstract}
Let $G_q$ be an alphabet of size $q\geq2$.  We determine the optimal
asymptotic rate of generalized covering codes $C\subseteq G_q^n$, whose
covering centers in $G_q^{t\times n}$ are constrained to the product form
$C^t$.  For every fixed integer $t\geq1$ and every $\rho\in[0,1]$, we
prove that
\[
\kappa_t(\rho,q)=
\begin{cases}
1-H_{q^t}(\rho),&0\leq\rho<1-q^{-t},\\
0,&1-q^{-t}\leq\rho\leq1,
\end{cases}
\]
where $\kappa_t(\rho,q)$ denotes the minimum asymptotic rate
$n^{-1}\log_q|C|$ among codes whose $t$-th covering radius is at most
$\rho n$, and $H_{q^t}$ is the $q^t$-ary entropy function.  When $q$ is a
prime power, we prove that the same formula holds under the additional
requirement that $C\leq\F_q^n$.  Thus, both the product-form
constraint and linearity are asymptotically cost-free: the resulting rate
is the ordinary sphere-covering rate over an alphabet of size $q^t$.
This extends the recent $t=2$ result of Elimelech and Schwartz for codes
without a linearity constraint and the classical $t=1$ result of Cohen
and Frankl for linear codes, thereby resolving both open problems posed
by Elimelech and Schwartz.

Our proofs are probabilistic and combine tools from information theory
and probabilistic combinatorics, including the method of types, Janson's
inequality, the second-moment method, and a structured alteration
argument. Direct applications of Janson’s inequality and the second-moment method are obstructed by strong dependencies among candidate error matrices, namely, the differences between target matrices and covering centers. We overcome this obstruction by restricting these errors to a balanced exact-type class contained in the metric ball and having the same exponential size as the full ball.  Standard
type-class estimates, together with Shearer's inequality, then give the
required bounds on the number of error-matrix pairs whose selected rows
have a prescribed difference.
\end{abstract}

\begin{IEEEkeywords}
generalized covering codes, generalized covering radius,
linear codes, method of types, Janson's inequality, second-moment method
\end{IEEEkeywords}

\section{Introduction}

A covering code in a metric space is a set of codewords such that every
point of the space lies within a prescribed distance of some codeword.
Covering codes are classical objects in coding theory and discrete
mathematics, with connections to a variety of applications
\cite{Astola09,Bierbrauer2008,Linderoth09,Micciancio2004}; see
\cite{cohen-covering-codes} for a comprehensive treatment.

In a \(q\)-ary Hamming space, the covering problem has a familiar
football-pool interpretation \cite{Hamalainen1995}.  Suppose that \(n\)
games are played, each with \(q\) possible outcomes.  A ticket gives one
prediction per game and wins if at least \(n-r\) predictions are correct.
The minimum number of tickets required to guarantee a win for every
possible outcome sequence is precisely the minimum size of a length-\(n\)
\(q\)-ary code with covering radius at most \(r\).  This problem has been
studied extensively in combinatorics
\cite{Fernandes1983,Kamps1967,Ostergard1996}.

Elimelech and Schwartz \cite{elimelech-schwartz} introduced the
higher-order version of this model.  In the order-\(t\) model, each game
consists of \(t\) subgames, each with \(q\) possible outcomes, while a
single ticket gives only one \(q\)-ary prediction per game.  An ordered
selection of \(t\) tickets, not necessarily distinct, provides predictions
for all subgames: the \(j\)-th ticket is used for the \(j\)-th subgame of
every game.  A game is predicted correctly when all its \(t\) subgames are
predicted correctly, and the order-\(t\) football-pool problem asks for the
minimum number of tickets needed to guarantee at least \(n-r\) correctly
predicted games.

Equivalently, an outcome is a \(t\times n\) matrix over a \(q\)-symbol
alphabet, while an ordered choice of \(t\) tickets from a code \(C\) is a
matrix whose rows belong to \(C\), hence an element of \(C^t\).  The
\(t\)-distance between two such matrices is the number of columns in which
they differ.  Thus, the \(t\)-th covering radius \(R_t(C)\) is the maximum
distance from an outcome matrix to \(C^t\).  The problem is therefore a
covering problem in a Hamming space over an alphabet of size \(q^t\), with
the restriction that the centers have the product form \(C^t\).  When
\(q\) is a prime power and \(C\leq\F_q^n\), every row must be chosen from
the same linear code \(C\).

There is also an extension-field interpretation.  After fixing an
\(\F_q\)-basis of \(\F_{q^t}\), the space \(\F_q^{t\times n}\) is
isometric to \(\F_{q^t}^n\), and \(C^t\) is mapped to the
\(\F_{q^t}\)-linear span of \(C\).  Thus, the linear problem asks whether
\(\F_{q^t}\)-linear codes that admit a generator matrix with entries in
\(\F_q\) can attain the ordinary sphere-covering rate.  We make this
equivalence precise in Section~\ref{sec:preliminaries}.

The generalized covering radius was originally introduced for linear
codes in connection with linear data-querying protocols, where it
characterizes the trade-off among access complexity, storage, and latency
and is related to generalized Hamming weights
\cite{elimelech-firer-schwartz,VKWei1991}.  It has since been studied for
Reed--Muller codes \cite{elimelech-wei-schwartz}, for several families of
BCH and binary cyclic codes
\cite{yohananov-schwartz2025,ozbudak-ozturk2026,xiong-yip2026,
ozbudak-ozturk-second2026,luo-et-al2026,
barouch-essayag-zabokritskiy2026}, and through a finite-geometric
framework for linear codes \cite{alfarano2026geometric}.  We focus here on
its optimal asymptotic rate.

For an integer \(Q\geq2\) and \(x\in[0,1]\), let
\[
H_Q(x):=x\log_Q(Q-1)-x\log_Qx-(1-x)\log_Q(1-x),
\]
where \(0\log_Q0=0\).  Let \(\kappa_t(\rho,q)\) denote the optimal
asymptotic rate of \(q\)-ary codes whose normalized \(t\)-th covering
radius is at most \(\rho\), and, when \(q\) is a prime power, let
\(\kappa_t^{\mathrm{Lin}}(\rho,q)\) denote the corresponding optimal rate
for linear codes over \(\F_q\).  Formal definitions are given in
Section~\ref{sec:preliminaries}.
We refer to the former as unrestricted codes: here unrestricted means only
that no linearity condition is imposed on $C$, while the covering centers
remain constrained to the product form $C^t$.

For \(t=1\), the classical sphere-covering bound and the
Stein--Lov\'asz set-covering theorem \cite{Stein1974,Lovasz1975} give
\[
\kappa_1(\rho,q)=
\begin{cases}
1-H_q(\rho),&0\leq\rho<1-q^{-1},\\
0,&1-q^{-1}\leq\rho\leq1.
\end{cases}
\]
When \(q\) is a prime power, Cohen and Frankl \cite{Cohen85} further
showed that the same asymptotic rate is attainable by linear codes over
\(\F_q\).

Elimelech and Schwartz \cite{elimelech-schwartz} determined the
unrestricted rate for \(t=2\):
\[
\kappa_2(\rho,q)=
\begin{cases}
1-H_{q^2}(\rho),&0\leq\rho<1-q^{-2},\\
0,&1-q^{-2}\leq\rho\leq1.
\end{cases}
\]
For the linear rate problem, Elimelech, Firer, and Schwartz proved the
sphere-covering lower bound, general upper bounds, and a sharper
probabilistic upper bound in the special case \(q=2\) and \(t=2\)
\cite[Theorem~22]{elimelech-firer-schwartz}.  Their probabilistic upper
bound for \(q=2\) and \(t=2\), however, did not close the gap with the
sphere-covering bound in the interior of the nontrivial range.

Elimelech and Schwartz \cite{elimelech-schwartz} posed two further
questions.  Problem~1 asks whether the unrestricted rate equals the
sphere-covering rate for every fixed order \(t\geq3\), whereas Problem~2
asks whether imposing linearity changes the optimal asymptotic rate over a
finite field for every fixed order \(t\geq2\).  We resolve both problems.

Our first result treats codes that are not necessarily linear.
\begin{theorem}\label{thm:general-rate}
For every integer \(q\geq2\), every fixed integer \(t\geq1\), and every
\(\rho\in[0,1]\),
\[
\kappa_t(\rho,q)=
\begin{cases}
1-H_{q^t}(\rho),&0\leq\rho<1-q^{-t},\\
0,&1-q^{-t}\leq\rho\leq1.
\end{cases}
\]
\end{theorem}

Theorem~\ref{thm:general-rate} says that, in the ordinary covering problem
on \(G_q^{t\times n}\), restricting the centers to the product family
\(C^t\) does not change the optimal asymptotic covering rate.

Our second result shows that imposing linearity is also asymptotically
cost-free.
\begin{theorem}\label{thm:linear-rate}
For every prime power \(q\), every fixed integer \(t\geq1\), and every
\(\rho\in[0,1]\),
\[
\kappa_t^{\mathrm{Lin}}(\rho,q)=\kappa_t(\rho,q)=
\begin{cases}
1-H_{q^t}(\rho),&0\leq\rho<1-q^{-t},\\
0,&1-q^{-t}\leq\rho\leq1.
\end{cases}
\]
\end{theorem}

In the extension-field formulation,
Theorem~\ref{thm:linear-rate} determines the optimal asymptotic rate under the constraint that an \(\F_{q^t}\)-linear
covering code admit a generator matrix with entries in \(\F_q\).
Neither linearity nor this subfield-generation constraint changes the
sphere-covering rate.

The proof also gives the following finite-length statement when
\(t\geq2\).  For every rational \(0<\theta<1-q^{-t}\) and every
sufficiently large length \(n\), there is a linear code
\(C_n\leq\F_q^n\) such that
\[
R_t(C_n)\leq\theta n,
\qquad
\dim_{\F_q}C_n\leq n\bigl(1-H_{q^t}(\theta)\bigr)+O(\log n).
\]

\subsection{Proof overview.}
The lower bounds in the nontrivial range of
Theorems~\ref{thm:general-rate} and~\ref{thm:linear-rate} follow from the
sphere-covering bound
\eqref{eq:lower}.  We next outline the proofs of the matching
upper bounds.

Both upper-bound proofs face the same obstacle.  For an auxiliary relative
radius $\theta$, the full metric ball $\mathcal B_t(\theta n)$ has the exponential size required to attain the sphere-covering bound, but it contains
exponentially large subfamilies of matrices that are zero on a fixed
nonempty set of rows.  These subfamilies produce too many pairs that agree
on selected rows. Janson’s inequality in the proof of Theorem~\ref{thm:general-rate} requires bounds on these row-agreement counts, whereas the second-moment argument for Theorem~\ref{thm:linear-rate} requires bounds on the number of pairs having every prescribed difference on the selected rows.  

The common remedy is to replace the full ball by the balanced
exact-type family $\mathcal E_n$ defined in Section~\ref{sec:common}.  This restriction preserves the optimal
exponential size of the full ball, while requiring every nonzero column
vector to occur equally often.  Lemma~\ref{lem:count} bounds the
relevant pair counts in terms of the entropies of the marginal
distributions on the selected rows, and
Lemma~\ref{lem:entropy-balance} uses Shearer's inequality to give the
normalized marginal-entropy bound needed in both proofs.  In the
unrestricted proof, the resulting estimate controls pairs that agree on
selected rows; in the linear proof, the same estimate controls pairs
having any prescribed difference on those rows.

\paragraph{The Unrestricted Proof}
We sample $t$ auxiliary codes independently and take their union.  Their
Cartesian product gives the independence across rows needed by Janson's
inequality, while their union retains the desired rate.  With the
row-agreement bounds, Janson's inequality makes the probability that a
fixed target matrix $\mathbf v\in G_q^{t\times n}$ is uncovered doubly
exponentially small, so a union bound gives complete coverage.

\paragraph{The Linear Proof}
We take $C=\Row(G)$ for a random generator matrix with
$n(1-H_{q^t}(\theta))+O(\log n)$ rows.  The row-difference bounds and the
second-moment method give a code that covers all but a vanishing fraction
of the target matrices.  A structured alteration (see
Lemma~\ref{lem:density-squaring}), based on the augmentation technique of
\cite[Proposition~21]{elimelech-firer-schwartz}, repeatedly squares the
uncovered density and completes the covering while adding only $O(\log n)$
dimensions.

Lastly, we want to emphasize that the unrestricted proof obtains complete
coverage directly, whereas the linear proof first obtains almost-complete
coverage and then uses alteration.

The rest of the paper is organized as follows.
Section~\ref{sec:preliminaries} gives the preliminaries, including the
extension-field interpretation and the sphere-covering lower bound.
Section~\ref{sec:common} constructs $\mathcal E_n$ and establishes the
counting and entropy estimates used in both proofs. Sections~\ref{sec:nonlinear-proof} and
\ref{sec:linear-proof} prove the unrestricted and linear results,
respectively, and Section~\ref{sec:conclusion} concludes the paper.

\section{Preliminaries}\label{sec:preliminaries}
For consistency, we follow the notation of \cite{elimelech-schwartz}.
Throughout, $[m]=\{1,\ldots,m\}$ for every positive integer $m$.
We denote by $G_q$ an abelian group of order $q\geq 2$, write $+$ for its
group operation, and write $0$ for its identity element.  The sets
$G_q^n$ and $G_q^{t\times n}$ consist, respectively, of length-$n$
vectors and $t\times n$ matrices with entries in $G_q$. Both of them are regarded
as abelian groups under entrywise addition.

Lowercase letters, such as
$v$, denote elements of $G_q$ (and hence scalars when $G_q$ is a field);
overlined lowercase letters, such as $\overline v$, denote vectors; and
bold lowercase letters, such as $\mathbf v$, denote matrices.  Matrices may
also be denoted by uppercase letters, such as $A$.  For
$\mathbf v\in G_q^{t\times n}$, we
write $\overline v_i$ for its \(i\)-th row and $\overline v^{(j)}$ for its
\(j\)-th column.  Thus, lower indices refer to rows, whereas parenthesized
superscripts refer to columns.  If $S\subseteq[t]$, then $\mathbf v_S$
denotes the submatrix of $\mathbf v$ formed by the rows indexed by $S$.
We use $0$, $\overline 0$, and $\mathbf 0$ for the zero group element,
the all-zero vector, and the all-zero matrix, respectively, with their
dimensions clear from context.

When $q$ is a prime power, we take $G_q=\F_q$.  For a matrix
$A$ over $\F_q$, we denote its row space and rank by
$\Row(A)$ and $\rank A$, respectively, and write
$\GL_t(q)$ for the group of invertible $t\times t$ matrices over
$\F_q$.  The notation $C\leq\F_q^n$ means that $C$ is an
$\F_q$-linear subspace of $\F_q^n$.

A nonempty subset $C\subseteq G_q^n$ is called a code, and its elements
are called codewords.  When $q$ is a prime power, such a code is called
linear if it is a linear subspace of $\F_q^n$.  In that case, $C$
is an $[n,k]_q$ linear code, where
$k=\dim_{\F_q}C=\log_q|C|$.

\paragraph{Generalized Covering Radius and Optimal Rates}
We now define the $t$-metric for the $t$-th covering radius.
\begin{definition}
Let $\mathbf v\in G_q^{t\times n}$ have columns
$\overline v^{(1)},\ldots,\overline v^{(n)}\in G_q^t$.  The $t$-weight of
$\mathbf v$ is
$\wt^{(t)}(\mathbf v):=
\left|\{j\in[n]\mid\overline v^{(j)}\neq\overline 0\}\right|$,
where $\overline 0$ denotes the zero vector in $G_q^t$.
For $\mathbf v,\mathbf u\in G_q^{t\times n}$, their $t$-distance is
$d^{(t)}(\mathbf v,\mathbf u):=\wt^{(t)}(\mathbf v-\mathbf u)$.
The $t$-ball of radius $r$ centered at $\mathbf v$ is
$B_r^{(t)}(\mathbf v):=\{\mathbf u\in G_q^{t\times n}\mid
d^{(t)}(\mathbf v,\mathbf u)\leq r\}$.
\end{definition}

We next define the $t$-th power and $t$-th covering radius of a code.

\begin{definition}
Let $C\subseteq G_q^n$ be a nonempty code and let $t$ be a positive integer.
The $t$-th power of $C$ is
\[
C^t
:=
\left\{
\begin{bmatrix}
\overline c_1\\
\vdots\\
\overline c_t
\end{bmatrix}
\in G_q^{t\times n}
\ \middle|\
\forall i\in[t],
\overline c_i\in C
\right\}.
\]
The $t$-th covering radius of $C$ is the covering radius of $C^t$ in
$G_q^{t\times n}$ with
respect to the $t$-metric:
\[
R_t(C):=\max_{\mathbf u\in G_q^{t\times n}}
\min_{\mathbf c\in C^t}d^{(t)}(\mathbf c,\mathbf u).
\]
\end{definition}

For $t=1$, $d^{(1)}$ is the usual Hamming distance and $R_1(C)$ is the
ordinary covering radius of $C$; see~\cite{cohen-covering-codes} for a
standard reference.  We are interested in the minimum cardinality of a
code $C\subseteq G_q^n$ satisfying $R_t(C)\leq r$.

\begin{definition}
For integers \(n,t\geq1\) and \(q\geq2\), and for \(0\leq r\leq n\), define
\[
k_t(n,r,q):=\min
\{\log_q|C|\mid \emptyset\neq C\subseteq G_q^n,\ R_t(C)\leq r\}.
\]
For \(\rho\in[0,1]\), the optimal asymptotic rate for unrestricted codes is
\[
\kappa_t(\rho,q):=\liminf_{n\to\infty}\frac{k_t(n,\rho n,q)}{n}.
\]

When \(q\) is a prime power, define in addition
\[
k_t^{\mathrm{Lin}}(n,r,q):=\min
\{\dim_{\F_q}C\mid C\leq\F_q^n,\ R_t(C)\leq r\}
\]
and
\[
\kappa_t^{\mathrm{Lin}}(\rho,q):=
\liminf_{n\to\infty}\frac{k_t^{\mathrm{Lin}}(n,\rho n,q)}{n}.
\]
Since every linear code is an unrestricted code,
\(\kappa_t(\rho,q)\leq\kappa_t^{\mathrm{Lin}}(\rho,q)\).
\end{definition}

Although the definitions use group notation, $d^{(t)}(\mathbf v,\mathbf u)$
is simply the number of columns in which $\mathbf v$ and $\mathbf u$
differ.  Consequently, $R_t(C)$, $k_t(n,r,q)$, and $\kappa_t(\rho,q)$
depend only on the alphabet size $q$, not on the chosen group structure.

Since the distance is integer-valued, a nonintegral value of $\rho n$
causes no ambiguity.

\paragraph{Extension-Field Interpretation}
When \(q\) is a prime power, the preceding matrix model has an equivalent
formulation in the ordinary Hamming space over the extension field
\(\F_{q^t}\); see \cite[Lemma~9]{elimelech-firer-schwartz}.  Fix an
\(\F_q\)-basis
\(\beta_1,\ldots,\beta_t\) of \(\F_{q^t}\) and consider the
\(\F_q\)-linear bijection
\[
\begin{aligned}
\Phi&:\F_q^{t\times n}\longrightarrow\F_{q^t}^n,\\
\Phi(\mathbf v)&:=\sum_{i=1}^t\beta_i\overline v_i,
\end{aligned}
\]
where \(\overline v_i\) is the \(i\)-th row of \(\mathbf v\).
A column of \(\mathbf v\) is zero if and only if the corresponding
coordinate of \(\Phi(\mathbf v)\) is zero.  Hence \(\Phi\) is an isometry
from \(d^{(t)}\) to the Hamming metric on \(\F_{q^t}^n\).  Under this
isometry, if \(C\leq\F_q^n\), \(k=\dim_{\F_q}C\), and
\(G\in\F_q^{k\times n}\) is a generator matrix of \(C\), then
\[
\Phi(C^t)=\operatorname{span}_{\F_{q^t}}(C)=\F_{q^t}^kG.
\]
A full-rank minor of \(G\) over \(\F_q\) remains nonzero over
\(\F_{q^t}\).  Consequently,
\[
\begin{aligned}
\dim_{\F_{q^t}}\operatorname{span}_{\F_{q^t}}(C)&=k,\\
R_t(C)&=R_1\bigl(\operatorname{span}_{\F_{q^t}}(C)\bigr).
\end{aligned}
\]
Thus, the linear problem is equivalently an ordinary covering problem over
\(\F_{q^t}\) in which the covering code must admit a generator matrix with
entries in \(\F_q\).

\paragraph{Entropy and the Volume Bound}
Throughout, entropies denoted by $H(\cdot)$, including joint and
conditional entropies, use logarithms to base $q$. By contrast, the logarithms in the
$Q$-ary entropy function $H_Q$ are taken to base $Q$.
For a probability distribution $P$ on a finite set, write
\[
H(P):=-\sum_a P(a)\log_q P(a).
\]

Since the $t$-weight counts nonzero columns, every $t$-ball of radius
$\rho n$ has cardinality
$|B_{\rho n}^{(t)}(\mathbf v)|
=\sum_{j=0}^{\lfloor\rho n\rfloor}\binom nj(q^t-1)^j$.
Thus, the standard entropy estimate for binomial sums gives, for fixed
$\rho\in[0,1]$,
\[
\left|B_{\rho n}^{(t)}(\mathbf v)\right|
=
\begin{cases}
q^{tnH_{q^t}(\rho)+o(n)},
&0\leq\rho\leq1-q^{-t},\\
q^{tn+o(n)},
&1-q^{-t}<\rho\leq1.
\end{cases}
\]

We next present the sphere-covering lower bound common to the unrestricted
and linear settings.
If $R_t(C)\leq\rho n$, then the $|C|^t$ balls centered at the elements of
$C^t$ cover $G_q^{t\times n}$.  Since these balls all have the same size,
for every $\mathbf v\in G_q^{t\times n}$ we have
$|C|^t\left|B_{\rho n}^{(t)}(\mathbf v)\right|\geq q^{tn}$.
Together with the volume estimate above, this gives the standard lower
bound
\begin{equation}\label{eq:lower}
	\kappa_t(\rho,q)\geq
	\begin{cases}
		1-H_{q^t}(\rho), & 0\leq \rho<1-q^{-t},\\
		0, & 1-q^{-t}\leq \rho\leq 1.
	\end{cases}
\end{equation}

When \(q\) is a prime power, the inclusion of the linear codes among all
codes and \eqref{eq:lower} also give
\begin{equation}\label{eq:volume-lower-bound}
\kappa_t^{\mathrm{Lin}}(\rho,q)\geq\kappa_t(\rho,q)
\geq1-H_{q^t}(\rho),\qquad 0\leq\rho<1-q^{-t}.
\end{equation}

\section{A Balanced Exact-Type Error Family}
\label{sec:common}

For a target matrix $\mathbf v\in G_q^{t\times n}$ and a center
$\mathbf c\in C^t$, we regard the difference
$\mathbf e=\mathbf v-\mathbf c$ as an error matrix.  The condition
$R_t(C)\leq\theta n$ is equivalent to requiring that every target matrix $\mathbf v $
admit a representation $\mathbf v =\mathbf c+\mathbf e$ with $ \mathbf c \in C^t$ and $\wt^{(t)}(\mathbf e)\leq\theta n$.  In other
words, the translates by the centers in $C^t$ of the $t$-metric ball
\[
\mathcal B_t(\theta n):=
\{\mathbf e\in G_q^{t\times n}:\wt^{(t)}(\mathbf e)\leq\theta n\}
\]
cover the entire space $G_q^{t\times n}$.

In both probabilistic constructions, rather than using the full ball
$\mathcal B_t(\theta n)$, we work with a highly structured error family
$\mathcal E_n\subseteq\mathcal B_t(\theta n)$.  For admissible lengths
$n$, every matrix in $\mathcal E_n$ has exactly $(1-\theta)n$ zero
columns, while each nonzero vector of $G_q^t$ occurs exactly
$\theta n/(q^t-1)$ times as a column.  Thus, $\mathcal E_n$ is a balanced
exact-type class consisting of matrices of $t$-weight exactly $\theta n$.
Requiring the translates $C^t+\mathcal E_n$ to cover the ambient space is
stronger than ordinary radius-$\theta n$ covering.  The advantage is that
$\mathcal E_n$ has the same exponential size as the full ball, while its
balanced structure bounds how many pairs have any prescribed difference
on selected rows.

We now formalize this choice.  Fix integers $q\geq2$ and $t\geq2$ and a
rational number
$0<\theta<1-q^{-t}$,
and put $h:=H_{q^t}(\theta)$.  Define a distribution $P_\theta$ on
$G_q^t$ by
\[
P_\theta(\overline 0)=1-\theta,
\qquad
P_\theta(\overline a)=\frac{\theta}{q^t-1}
\quad(\overline a\neq\overline 0).
\]
Its entropy is
\begin{equation}\label{eq:full-entropy}
\begin{aligned}
H(P_\theta)
&=-(1-\theta)\log_q(1-\theta)
-\theta\log_q\frac{\theta}{q^t-1}\\
&=tH_{q^t}(\theta)=th.
\end{aligned}
\end{equation}

Call a positive integer $n$ admissible if
$nP_\theta(\overline a)$ is an integer for every
$\overline a\in G_q^t$.  Since $\theta$ is rational, every sufficiently
large multiple of a fixed positive integer is admissible.  For such $n$,
define the exact column-type class
\[
\mathcal E_n:=\left\{\mathbf e\in G_q^{t\times n}:
\left|\{j\in[n]:\overline e^{(j)}=\overline a\}\right|
=nP_\theta(\overline a)\text{ for every }\overline a\in G_q^t\right\}.
\]
Every $\mathbf e\in\mathcal E_n$ has exactly $\theta n$ nonzero columns.
When $q$ is a prime power,
the class is invariant under left multiplication by every $U\in\GL_t(q)$,
because $U$ fixes zero and permutes the nonzero vectors of $\F_q^t$.

Let $X=(X_1,\ldots,X_t)\sim P_\theta$.  For every nonempty
$S\subseteq[t]$, let $P_{\theta,S}$ be the distribution of $X_S$.

For a finite alphabet $\mathcal A$ and a sequence
$\overline x=(x_1,\ldots,x_m)\in\mathcal A^m$, its type, or empirical
distribution, is the probability distribution
\[
\widehat P_{\overline x}(a)
:=\frac1m\bigl|\{j\in[m]:x_j=a\}\bigr|
\qquad(a\in\mathcal A).
\]
A probability distribution $Q$ on $\mathcal A$ is called an $m$-type if
$mQ(a)$ is an integer for every $a\in\mathcal A$.  Its type class is
\[
\mathcal T_m(Q):=
\{\overline x\in\mathcal A^m:\widehat P_{\overline x}=Q\}.
\]
After identifying a matrix with its ordered sequence of columns,
$\mathcal E_n$ is precisely $\mathcal T_n(P_\theta)$ on the alphabet
$G_q^t$.
We use the following standard type-class inequality (see \cite[Th.~11.1.3]{cover-thomas} and
\cite[Ch.~2, Lem.~2.3]{csiszar-korner}):
\begin{equation}\label{eq:standard-type-class}
(m+1)^{-|\mathcal A|}q^{mH(Q)}\leq |\mathcal T_m(Q)|\leq q^{mH(Q)}.
\end{equation}

For a nonempty $S\subseteq[t]$, put $s:=|S|$.  For an admissible $n$ and
$\mathbf x\in G_q^{s\times n}$, set
\[
m_S(\mathbf x):=\left|\{\mathbf e\in\mathcal E_n:\mathbf e_S=\mathbf x\}\right|.
\]
For any family $\mathcal E\subseteq G_q^{t\times n}$ and any
$\mathbf z\in G_q^{s\times n}$, define
\[
N_S(\mathcal E;\mathbf z):=
\left|\{(\mathbf e,\mathbf f)\in\mathcal E^2:
\mathbf f_S-\mathbf e_S=\mathbf z\}\right|.
\]
We call this the row-difference count at $\mathbf z$.  In particular,
$N_S(\mathcal E;\mathbf 0)$ is the row-agreement count used in the
unrestricted proof, whereas the linear covariance calculation requires
the bound for every $\mathbf z$.
The following two lemmas give the bounds on $m_S$ and $N_S$ used in both
constructions.

\begin{lemma}
\label{lem:count}
For every admissible $n$,
\begin{equation}
\label{eq:type-class-bounds}
(n+1)^{-q^t}q^{nH(P_\theta)}
\leq|\mathcal E_n|\leq q^{nH(P_\theta)}.
\end{equation}
Moreover, for every nonempty $S\subseteq[t]$ and every admissible $n$,
\begin{align}
m_S(\mathbf x)
&\leq q^{n(H(P_\theta)-H(P_{\theta,S}))}
\quad\text{for every }\mathbf x\in G_q^{|S|\times n},
\label{eq:prescribed-rows-count}\\
N_S(\mathcal E_n;\mathbf z)
&\leq q^{n(2H(P_\theta)-H(P_{\theta,S}))}
\quad\text{for every }\mathbf z\in G_q^{|S|\times n}.
\label{eq:type-class-row-difference-count}
\end{align}
\end{lemma}

\begin{proof}
Applying \eqref{eq:standard-type-class} to the column type $P_\theta$
gives \eqref{eq:type-class-bounds}.

Fix $S$ and $\mathbf x$.  If $m_S(\mathbf x)>0$, then $\mathbf x$ has
column type $P_{\theta,S}$.  For every
$\overline b\in G_q^{|S|}$, the positions in which the column of
$\mathbf x$ equals $\overline b$ must be filled by columns
$\overline a\in G_q^t$ satisfying $\overline a_S=\overline b$, using
each $\overline a$ exactly $nP_\theta(\overline a)$ times.  Hence
\[
m_S(\mathbf x)
=\prod_{\overline b\in G_q^{|S|}}
\frac{\bigl(nP_{\theta,S}(\overline b)\bigr)!}
{\displaystyle
\prod_{\substack{\overline a\in G_q^t:\\
\overline a_S=\overline b}}
\bigl(nP_\theta(\overline a)\bigr)!}.
\]
For each $\overline b$, the corresponding factor is a type class whose
type is the conditional distribution of $X$ given
$X_S=\overline b$.  Applying the type-class upper bound to every factor
and multiplying gives
\[
m_S(\mathbf x)
\leq
q^{n\sum_{\overline b}P_{\theta,S}(\overline b)
H(X\mid X_S=\overline b)}
=q^{nH(X\mid X_S)}
=q^{n(H(P_\theta)-H(P_{\theta,S}))},
\]
which proves~\eqref{eq:prescribed-rows-count}.  Finally,
\[
N_S(\mathcal E_n;\mathbf z)
=\sum_{\mathbf x\in G_q^{|S|\times n}}
m_S(\mathbf x)m_S(\mathbf x+\mathbf z)
\leq\left(\max_{\mathbf x}m_S(\mathbf x)\right)|\mathcal E_n|.
\]
Combining \eqref{eq:prescribed-rows-count} with the upper bound in
\eqref{eq:type-class-bounds} proves~\eqref{eq:type-class-row-difference-count}.
\end{proof}

In particular, \eqref{eq:full-entropy} and
\eqref{eq:type-class-bounds} give
\begin{equation}\label{eq:balanced-size}
|\mathcal E_n|=q^{nth+O(\log n)}.
\end{equation}

The distribution $P_\theta$ is invariant under coordinate permutations.
The following lemma follows from Shearer's inequality
\cite{chung-graham-frankl-shearer} and exchangeability.  We include a short
direct proof for completeness.

\begin{lemma}
\label{lem:entropy-balance}
For every nonempty subset $S\subseteq[t]$, with $s=|S|$,
\[
H(P_{\theta,S})\geq sH_{q^t}(\theta)=sh.
\]
\end{lemma}

\begin{proof}
Let $\binom{[t]}s$ denote the family of all $s$-subsets of $[t]$.
For $A\in\binom{[t]}s$, the entropy chain rule and the fact that
conditioning cannot increase entropy give
\[
H(X_A)\geq\sum_{i\in A}H(X_i\mid X_1,\ldots,X_{i-1}).
\]
Summing over $A\in\binom{[t]}s$ and observing that every $i\in[t]$
belongs to $\binom{t-1}{s-1}$ such subsets yields the specialization of
Shearer's inequality to all $s$-subsets of $[t]$; see
\cite[p.~33, Eq.~(22)]{chung-graham-frankl-shearer}:
\[
\sum_{A\in\binom{[t]}s}H(X_A)\geq\binom{t-1}{s-1}H(X).
\]
By exchangeability, the left-hand side equals
$\binom tsH(P_{\theta,S})$.  Since $H(X)=th$ by
\eqref{eq:full-entropy}, it follows that
\[
H(P_{\theta,S})\geq
\frac{\binom{t-1}{s-1}}{\binom ts}H(X)=sh.
\]
\end{proof}

Combining Lemmas~\ref{lem:count} and~\ref{lem:entropy-balance}, for every
nonempty proper subset $S\subsetneq[t]$, with $s=|S|$, and every
$\mathbf z\in G_q^{s\times n}$, we obtain
\begin{equation}\label{eq:balanced-row-difference-count}
N_S(\mathcal E_n;\mathbf z)
\leq q^{n(2t-s)H_{q^t}(\theta)}=q^{n(2t-s)h}.
\end{equation}

\section{The Unrestricted Case}
\label{sec:nonlinear-proof}

By~\eqref{eq:lower}, only the matching upper bound in
Theorem~\ref{thm:general-rate} remains to be proved.  The cases \(t=1,2\)
are known \cite{Stein1974,Lovasz1975,elimelech-schwartz}, so throughout
this section we fix an integer \(t\geq3\).  We first show that the full
metric ball can make even the extended Janson
bound ineffective.  We then use the balanced exact-type family constructed
in Section~\ref{sec:common} to overcome this obstruction.

\subsection{The Janson Setup and the Row-Agreement Obstruction}
\label{subsec:janson-obstruction}

We use the following combined consequence of the ordinary and extended
forms of Janson's inequality; see
\cite[Theorems~8.1.1 and~8.1.2]{alon-spencer}.

\begin{lemma}[Janson's inequality]
\label{lem:janson}
Let $\Omega$ be a finite set, and let $R$ be a random subset of $\Omega$
obtained by including its elements mutually independently, possibly with
different probabilities.  Let $\mathcal H$ be a finite family of subsets
of $\Omega$.  For every $E\in\mathcal H$, set
$I_E:=\ind_{\{E\subseteq R\}}$ and
$X:=\sum_{E\in\mathcal H}I_E$.  Let $\mu:=\Exp[X]$ and
\[
\Delta:=
\sum_{\substack{E,F\in\mathcal H,\ E\neq F\\
E\cap F\neq\emptyset}}
\Exp[I_EI_F],
\]
where the latter sum is over ordered pairs.  If $\mu>0$, then
\begin{equation}
\label{eq:combined-janson}
\Prob[X=0]\leq\exp\left(-\frac{\mu^2}{2(\mu+\Delta)}\right).
\end{equation}
\end{lemma}

This follows by using the ordinary Janson inequality when $\Delta<\mu$
and the extended Janson inequality when $\Delta\geq\mu$.

To describe the construction in the nontrivial range, fix
$0<\rho<1-q^{-t}$ and a rational number $0<\rho_0<\rho$.  Here $\rho$ is
the desired normalized covering radius, while the candidate error
matrices will be restricted to have $t$-weight at most $\rho_0n$.  In
Subsection~\ref{subsec:completion}, $\rho_0$ will be chosen sufficiently
close to $\rho$.  Fix also a real number
$0<\eta<H_{q^t}(\rho_0)$, and set
\begin{equation}\label{eq:setprobability}
p:=q^{-n(H_{q^t}(\rho_0)-\eta)}.
\end{equation}
Independently choose random subsets $C_1,\ldots,C_t\subseteq G_q^n$ by
including every vector in every $C_i$ with probability $p$, with all
membership choices mutually independent, and set
$C:=C_1\cup\cdots\cup C_t$.

Fix a target matrix $\mathbf v\in G_q^{t\times n}$ with rows
$\overline v_1,\ldots,\overline v_t$, and let
$\mathcal E\subseteq G_q^{t\times n}$ be a family of error matrices, each
of $t$-weight at most $\rho_0n$.  On the ground set $[t]\times G_q^n$, define
\[
R:=\{(i,\overline x)\in[t]\times G_q^n\mid\overline x\in C_i\}.
\]
For $\mathbf e\in\mathcal E$ with rows
$\overline e_1,\ldots,\overline e_t$, set
\[
E_{\mathbf e}:=\{(i,\overline v_i-\overline e_i)\mid i\in[t]\} 
\quad \textrm{and} \quad I_{\mathbf e}:=\ind_{\{E_{\mathbf e}\subseteq R\}}.\]
Let 
$$\mathcal H_{\mathbf v}:= \{E_{\mathbf e}\mid \mathbf e\in\mathcal E \}\quad  \textrm{and} \quad X_{\mathbf v}:=\sum_{\mathbf e\in\mathcal E}I_{\mathbf e}.$$
If $X_{\mathbf v}>0$, then
$\mathbf v-\mathbf e\in C_1\times\cdots\times C_t\subseteq C^t$ for some
$\mathbf e\in\mathcal E$, so $\mathbf v$ is covered within distance
$\rho_0n$.  Moreover,
\begin{equation}\label{eq:janson-mean-general}
\mu:=\Exp[X_{\mathbf v}]=|\mathcal E|p^t.
\end{equation}
Denote the corresponding dependency sum in Lemma~\ref{lem:janson} by
$\Delta_{\mathbf v}$.

For each nonempty $S\subseteq[t]$, consider the row-agreement count
$N_S(\mathcal E;\mathbf 0)$ defined in Section~\ref{sec:common}.
Take two distinct error matrices
$\mathbf e,\mathbf f\in\mathcal E$, with associated 
$E_{\mathbf e}, E_{\mathbf f}\in \mathcal H_{\mathbf v}$.  Suppose that they
agree precisely on the rows indexed by a nonempty set
$S\subsetneq[t]$; that is,
$\overline e_i=\overline f_i$ if and only if $i\in S$.  Writing
$s=|S|$, we then have
\[
E_{\mathbf e}\cap E_{\mathbf f}
=\{(i,\overline v_i-\overline e_i)\mid i\in S\},
\qquad |E_{\mathbf e}\cap E_{\mathbf f}|=s.
\]
Consequently,
$\Exp[I_{\mathbf e}I_{\mathbf f}]
=\Prob[E_{\mathbf e}\cup E_{\mathbf f}\subseteq R]=p^{2t-s}$.
Grouping the ordered pairs according to their common rows
therefore gives the bound
\begin{equation}\label{eq:janson-dependency-general}
\Delta_{\mathbf v}\leq
\sum_{\emptyset\neq S\subsetneq[t]}
N_S(\mathcal E;\mathbf 0)p^{2t-|S|}.
\end{equation}
This upper bound includes the diagonal pairs in
$N_S(\mathcal E;\mathbf 0)$, although they do not occur in
$\Delta_{\mathbf v}$, and it also overcounts pairs that agree on rows
outside $S$.  Indeed, a distinct pair whose exact common-row set is $T$
is counted in $N_S(\mathcal E;\mathbf 0)$ for every nonempty
$S\subseteq T$.  These extra nonnegative terms only enlarge the
right-hand side of~\eqref{eq:janson-dependency-general}.

To see the obstruction for the natural unstructured family, take
$\mathcal E$ to be the entire $t$-ball of admissible errors:
\[
\mathcal E=\mathcal B_t(\rho_0n).
\]
For lengths such that $\rho_0n$ is an integer, its cardinality has the
sphere-volume exponent, i.e.,
\[
|\mathcal B_t(\rho_0n)|=q^{ntH_{q^t}(\rho_0)+o(n)},
\qquad \mu=|\mathcal B_t(\rho_0n)|p^t=q^{nt\eta+o(n)}.
\]
The problem is that the ball also contains exponentially large subfamilies
supported only on a proper subset of the rows.  Fix a nonempty proper
subset $S\subsetneq[t]$, write $s=|S|$, and let
\[
\mathcal B_S:=\{\mathbf e\in\mathcal B_t(\rho_0n)\mid\mathbf e_S=\mathbf 0\}.
\]
Thus, $\mathcal B_S$ consists of matrices supported only on the rows
outside $S$.
For $\rho_0<1-q^{-(t-s)}$, the remaining $t-s$ rows form a ball, so
\[
|\mathcal B_S|=q^{n(t-s)H_{q^{t-s}}(\rho_0)+o(n)}.
\]
For any two distinct errors $\mathbf e,\mathbf f\in\mathcal B_S$, the
associated $E_{\mathbf e}$ and $ E_{\mathbf f}$ share at least the $s$ elements
$\{(i,\overline v_i)\mid i\in S\}$.  If they agree on exactly $r\geq s$
rows, then
\[
\Exp[I_{\mathbf e}I_{\mathbf f}]=p^{2t-r}\geq p^{2t-s},
\]
where the inequality uses $0<p<1$.  Restricting the ordered-pair sum
defining $\Delta_{\mathbf v}$ to the distinct ordered pairs from
$\mathcal B_S$ therefore gives
\[
\Delta_{\mathbf v}
\geq
\sum_{\substack{\mathbf e,\mathbf f\in\mathcal B_S\\
\mathbf e\neq\mathbf f}}
\Exp[I_{\mathbf e}I_{\mathbf f}]
\geq |\mathcal B_S|(|\mathcal B_S|-1)p^{2t-s}.
\]
Consequently,
\[
\frac1n\log_q\frac{\Delta_{\mathbf v}}{\mu^2}
\geq 2(t-s)H_{q^{t-s}}(\rho_0)
-(2t-s)H_{q^t}(\rho_0)-s\eta+o(1).
\]
Put $D_s(x):=2(t-s)H_{q^{t-s}}(x)-(2t-s)H_{q^t}(x)$.
As $x\to0$,
$D_s(x)=\frac{s}{t}x\log_q\frac1x+O(x)$.
Thus $D_s(x)>0$ for every sufficiently small positive $x$.  Fix such a
$\rho_0$ and choose $0<\eta<D_s(\rho_0)/(2s)$.  Then
\[
\frac{\Delta_{\mathbf v}}{\mu^2}
\geq q^{nD_s(\rho_0)/2+o(n)},\qquad
\frac{\mu^2}{\mu+\Delta_{\mathbf v}}
\leq \frac{\mu^2}{\Delta_{\mathbf v}}\longrightarrow0.
\]

Therefore, even the extended Janson bound is ineffective for the full
ball in this range.  Since attaining the sharp rate requires allowing
$\eta$ to be arbitrarily small, this is a genuine obstruction to using
the unstructured ball throughout the full parameter range.  The
obstruction comes from exponentially large subfamilies whose matrices are
zero on a fixed nonempty set of rows, not from a loss in the ball's total
cardinality.  We therefore retain the optimal exponential cardinality but
restrict the candidate errors to the balanced exact-type family
$\mathcal E_n$.

\subsection{Completion of the Proof}
\label{subsec:completion}

\begin{proof}[Proof of Theorem~\ref{thm:general-rate}]
Since the cases $t=1,2$ are already known, it remains to prove the upper
bound for the fixed $t\geq3$.

First suppose that $\rho=0$.  The condition $R_t(C)=0$ implies
$C^t=G_q^{t\times n}$ and hence $C=G_q^n$.  Thus
$k_t(n,0,q)=n$ and $\kappa_t(0,q)=1$.

Now let $0<\rho<1-q^{-t}$ and fix $\delta>0$.  Choose a rational number
$\rho_0$ such that $0<\rho_0<\rho$ and
$H_{q^t}(\rho_0)>H_{q^t}(\rho)-\delta/2$.
Such a choice is possible by the continuity and strict monotonicity of
$H_{q^t}$ on $[0,1-q^{-t}]$.
Use the family $\mathcal E_n$ from Section~\ref{sec:common} with
$\theta=\rho_0$.  Choose
\[
0<\eta<\min\left\{\frac{\delta}{4},
\frac{H_{q^t}(\rho_0)}{2}\right\}.
\]
The two upper bounds respectively control the loss in the final rate and
ensure that $0<p<1$.

For every admissible value of $n$, perform
the random construction from Subsection~\ref{subsec:janson-obstruction}
with $\mathcal E=\mathcal E_n$.  By~\eqref{eq:setprobability}, for every
$i\in[t]$, we have
$\Exp[|C_i|]=q^np=q^{n(1-H_{q^t}(\rho_0)+\eta)}$.
Since
$|C|\leq\sum_{i=1}^t|C_i|$, we have
$\Exp[|C|]\leq tq^{n(1-H_{q^t}(\rho_0)+\eta)}$.
Markov's inequality gives
  \begin{equation}\label{eq:code-size-probability}
  \Prob\left[|C|>q^{n(1-H_{q^t}(\rho_0)+2\eta)}\right]
  \leq\frac{\sum_{i=1}^t\Exp[|C_i|]}
  {q^{n(1-H_{q^t}(\rho_0)+2\eta)}}=tq^{-n\eta}.
  \end{equation}

For a fixed target matrix $\mathbf v$, the estimate for $|\mathcal E_n|$ in
\eqref{eq:balanced-size}, together
with~\eqref{eq:janson-mean-general}, gives
\[
\mu=|\mathcal E_n|p^t=q^{nt\eta+O(\log n)}.
\]
Taking $\mathbf z=\mathbf 0$ in
\eqref{eq:balanced-row-difference-count} and using
\eqref{eq:janson-dependency-general}, we obtain
\[
\Delta_{\mathbf v}
\leq\sum_{s=1}^{t-1}\binom tsq^{n(2t-s)\eta+O(\log n)}
\leq q^{n(2t-1)\eta+O(\log n)}.
\]
These two estimates imply
\begin{equation}
\label{eq:janson-ratio}
\frac{\mu^2}{\mu+\Delta_{\mathbf v}}\geq q^{n\eta-O(\log n)}.
\end{equation}
Indeed, the numerator is $q^{2nt\eta+O(\log n)}$, while the denominator
is at most $q^{n(2t-1)\eta+O(\log n)}$.  By~\eqref{eq:janson-ratio},
  Janson's inequality~\eqref{eq:combined-janson}
  therefore yields, for every
  $\mathbf v\in G_q^{t\times n}$,
  \[
  \Prob[X_{\mathbf v}=0]
  \leq\exp\left(-\frac{\mu^2}{2(\mu+\Delta_{\mathbf v})}\right)
  \leq\exp\left(-q^{n\eta-O(\log n)}\right).
  \]

There are $q^{tn}$ target matrices.  Hence the union bound gives
\[
\Prob\left[
\text{there exists }\mathbf v\in G_q^{t\times n}\text{ with }X_{\mathbf v}=0
\right]
\leq
q^{tn}\exp\left(-q^{n\eta-O(\log n)}\right)
\longrightarrow0.
\]
If $X_{\mathbf v}>0$ for every target matrix $\mathbf v$, then every
target matrix is covered within distance $\rho_0n<\rho n$.  The probability
of this event
failing and the probability in~\eqref{eq:code-size-probability} both tend
to zero, so their sum is less than $1$ for every sufficiently large
  admissible value of $n$.  Hence
  there exists a code $C_n\subseteq G_q^n$ such that
  \[
  R_t(C_n)\leq\rho n,
  \qquad |C_n|\leq q^{n(1-H_{q^t}(\rho_0)+2\eta)}.
  \]
Consequently,
\[
\frac{k_t(n,\rho n,q)}{n}\leq1-H_{q^t}(\rho_0)+2\eta
<1-H_{q^t}(\rho)+\delta.
\]
The admissible values of $n$ form an unbounded subsequence, which is
sufficient because $\kappa_t(\rho,q)$ is defined by a liminf over all
lengths.  Taking the liminf along this subsequence and then letting
$\delta\to0$ yields $\kappa_t(\rho,q)\leq1-H_{q^t}(\rho)$.

Finally, suppose that $1-q^{-t}\leq\rho\leq1$.  For every $\delta>0$,
choose $0<\rho'<1-q^{-t}$ such that
$1-H_{q^t}(\rho')<\delta$.
Such a choice is possible because
$H_{q^t}(x)\to1$ as $x\to(1-q^{-t})^-$.  Since
$\kappa_t(\rho,q)$ is nonincreasing in $\rho$, the upper bound just
proved gives
\[
0\leq\kappa_t(\rho,q)\leq\kappa_t(\rho',q)
\leq1-H_{q^t}(\rho')<\delta.
\]
Letting $\delta\to0$ gives $\kappa_t(\rho,q)=0$, completing the proof.
\end{proof}

\section{The Linear Case}
\label{sec:linear-proof}

Suppose that \(q\) is a prime power.  For \(t=1\), the assertion is the
classical theorem of Cohen and Frankl \cite{Cohen85}; hence the new content
is the case \(t\geq2\).  Fix a rational number
\(0<\theta<1-q^{-t}\), put \(h:=H_{q^t}(\theta)\), and use the family
\(\mathcal E_n\) from Section~\ref{sec:common}.  Recall that
\(\mathcal E_n\) is invariant under left multiplication by every
\(U\in\GL_t(q)\).

In the linear-code setting, the covariance calculation requires bounds on
pairs of errors for which the difference on selected rows is prescribed.
The bound \eqref{eq:balanced-row-difference-count} applies with
$G_q=\F_q$ for every nonempty proper subset $S\subsetneq[t]$ and every
$\mathbf z\in\F_q^{|S|\times n}$.  Indeed, the proof of
Lemma~\ref{lem:count} uses only column-type counting and does not use the
field structure.  Thus, the unrestricted proof takes
$\mathbf z=\mathbf 0$, whereas the linear proof must allow every
$\mathbf z$.

\subsection{Random Linear Codes Cover Almost All Target Matrices}
In this subsection, we show that a random linear code can cover almost all
target matrices with the error family $\mathcal E_n$.  Assume $t\geq2$,
fix an admissible $n$, and let $L=L(q,t)$ be a sufficiently
large constant.  Set
\begin{equation}\label{eq:dimension-choice}
k:=\left\lceil n(1-h)+L\log_q(n+1)\right\rceil.
\end{equation}
Put
\begin{equation}\label{eq:dimension-slack}
a_n:=k-n(1-h),\qquad
L\log_q(n+1)\leq a_n<L\log_q(n+1)+1.
\end{equation}
For all sufficiently large $n$, we have $2t\leq k<n$.  Choose
$G\in\F_q^{k\times n}$ by taking all its entries independently and
uniformly from $\F_q$, and let $C=\Row(G)$.  We do not condition on $G$
having full row rank, since $\dim_{\F_q}C\leq k$ is sufficient.

We count only representations whose coefficient matrices have full row
rank.  This restriction can only make coverage harder, while retaining
$q^{tk+O(1)}$ candidates.   Let
\[
\mathcal A_k:=\{A\in\F_q^{t\times k}:\rank A=t\}.
\]
Then for each $A\in \mathcal A_k
$,  the map
$G\mapsto AG$ is surjective, and hence \(AG\) is uniformly distributed over
\(\F_q^{t\times n}\). 
For a target matrix $\mathbf v\in\F_q^{t\times n}$, define
\[
X_{\mathbf v}:=\sum_{A\in\mathcal A_k}
\ind_{\{\mathbf v-AG\in\mathcal E_n\}}.
\]
If $X_{\mathbf v}>0$, then $\mathbf v$ lies at $t$-distance $\theta n$
from a center in $C^t$.

The next proposition shows that the probability of failure tends to zero
for every target matrix, with a bound independent of $\mathbf v$.

\begin{proposition}\label{prop:almost-covering}
If $L$ in \eqref{eq:dimension-choice} is sufficiently large, then for all
sufficiently large admissible $n$ and every
$\mathbf v\in\F_q^{t\times n}$, we have
\[
\mu:=\Exp[X_{\mathbf v}]\geq(n+1)^2,
\qquad
\frac{\operatorname{Var}(X_{\mathbf v})}{\mu^2}=O((n+1)^{-2}),
\]
where the implied constant is independent of $\mathbf v$.  Consequently,
$\Prob[X_{\mathbf v}=0]=O((n+1)^{-2})$.
\end{proposition}

\begin{proof}
For each $A\in\mathcal A_k$, the map $G\mapsto AG$ is a surjective linear
map from $\F_q^{k\times n}$ to $\F_q^{t\times n}$.  Hence $AG$ is uniform
on $\F_q^{t\times n}$.  Put
\[
p_n:=\frac{|\mathcal E_n|}{q^{tn}}=\Prob[\mathbf v-AG\in\mathcal E_n].
\]
Moreover, by \cite[Equation~1]{fulman-goldstein2015}, we have
\[
|\mathcal A_k|=\prod_{i=0}^{t-1}(q^k-q^i)\geq c_{q,t}q^{tk}
\]
for all sufficiently large $n$, where $c_{q,t}>0$.  The lower type-class
bound in \eqref{eq:type-class-bounds} now gives
\[
\mu=|\mathcal A_k|p_n\geq c_{q,t}(n+1)^{-q^t}q^{ta_n}.
\]
By~\eqref{eq:dimension-slack}, choosing $L$ sufficiently large
makes $\mu\geq(n+1)^2$ and also
\begin{equation}\label{eq:inverse-mean}
\mu^{-1}=O((n+1)^{-2}).
\end{equation}

We now turn to the variance.  Write
$I_A:=\ind_{\{\mathbf v-AG\in\mathcal E_n\}}$.  Then
\[
\operatorname{Var}(X_{\mathbf v})
=\sum_A\operatorname{Var}(I_A)
+\sum_{A\neq B}\bigl(\Exp[I_AI_B]-p_n^2\bigr).
\]
Since each $I_A$ is Bernoulli with mean $p_n$, the diagonal sum satisfies
\[
\sum_A\operatorname{Var}(I_A)=|\mathcal A_k|p_n(1-p_n)
\leq|\mathcal A_k|p_n=\mu.
\]

It remains to estimate the off-diagonal sum, which we classify according to
the intersection dimension.
For $A,B\in\mathcal A_k$, set
\[
s:=\dim\bigl(\Row(A)\cap\Row(B)\bigr).
\]
If $s=0$, the stacked matrix with row blocks $A$ and $B$ has rank $2t$.
Therefore $(AG,BG)$ is uniform on $\F_q^{2t\times n}$, so $I_A$ and
$I_B$ are independent and $\Exp[I_AI_B]-p_n^2=0$.

If $s=t$ and $A\neq B$, then $B=UA$ for a unique
$U\in\GL_t(q)\setminus\{I_t\}$, where $I_t$ denotes the $t\times t$
identity matrix.  Conversely, every such $U$ gives one
such $B$, so for each fixed $A$ there are exactly $|\GL_t(q)|-1$ choices
for $B$.  Moreover, $I_AI_B\leq I_A$, and hence
$\Exp[I_AI_B]-p_n^2\leq\Exp[I_AI_B]\leq p_n$.  Therefore the $s=t$ part
of the covariance sum satisfies
\[
\begin{aligned}
\sum_{\substack{A,B\in\mathcal A_k,\ A\neq B\\
                  \Row(A)=\Row(B)}}
\bigl(\Exp[I_AI_B]-p_n^2\bigr)
&\leq |\mathcal A_k|\bigl(|\GL_t(q)|-1\bigr)p_n\\
&=\bigl(|\GL_t(q)|-1\bigr)\mu=O(\mu).
\end{aligned}
\]

Consider now $1\leq s\leq t-1$.  Choose an ordered basis of
$\Row(A)\cap\Row(B)$ and extend it separately to ordered bases of
$\Row(A)$ and $\Row(B)$.  Since the rows of $A$ and $B$ are bases of
their respective row spaces, there exist $U,V\in\GL_t(q)$ that transform
the original row bases into these two extended bases.  In particular, the
first $s$ rows of $UA$ and $VB$ are the same chosen basis of the
intersection.  These common rows, together with the remaining $t-s$ rows
of each of $UA$ and $VB$, form a full-row-rank matrix
$M\in\F_q^{(2t-s)\times k}$, where the common rows are listed first,
followed by the remaining rows of $UA$ and then those of $VB$.  Hence
$MG$ is uniform on $\F_q^{(2t-s)\times n}$.

Consider the joint probability
\[
\Prob[\mathbf v-AG\in\mathcal E_n,\ \mathbf v-BG\in\mathcal E_n].
\]
Since $\mathcal E_n$ is invariant under $\GL_t(q)$, the event in this
probability is equivalent to the existence of
$\mathbf e,\mathbf f\in\mathcal E_n$ satisfying
\[
U\mathbf v-UAG=\mathbf e,
\qquad V\mathbf v-VBG=\mathbf f.
\]
On the $s$ common rows, the two errors satisfy
\[
\mathbf f_{[s]}-\mathbf e_{[s]}=\mathbf z,
\qquad
\mathbf z:=(V\mathbf v)_{[s]}-(U\mathbf v)_{[s]}.
\]
This condition is precisely what makes the
first $s$ rows of $U\mathbf v-\mathbf e$ and
$V\mathbf v-\mathbf f$ coincide.  Thus
each pair $(\mathbf e,\mathbf f)\in\mathcal E_n^2$ satisfying this condition
determines the value
\[
MG=
\begin{pmatrix}
(U\mathbf v-\mathbf e)_{[s]}\\
(U\mathbf v-\mathbf e)_{[t]\setminus[s]}\\
(V\mathbf v-\mathbf f)_{[t]\setminus[s]}
\end{pmatrix}.
\]
Conversely, a value of $MG$ for which the joint event occurs uniquely
determines $(\mathbf e,\mathbf f)$.  Therefore the admissible pairs are in
bijection with the values of $MG$ for which the joint event occurs.
Consequently,
\eqref{eq:balanced-row-difference-count} gives
\[
\Prob[\mathbf v-AG\in\mathcal E_n,\ \mathbf v-BG\in\mathcal E_n]
=\frac{N_{[s]}(\mathcal E_n;\mathbf z)}{q^{(2t-s)n}}
\leq q^{n((2t-s)h-(2t-s))}.
\]

The number of ordered pairs $(A,B)\in\mathcal A_k^2$ whose row spaces intersect in dimension
$s$ is at most
\[
c'_{q,t}q^{(2t-s)k}.
\]
To see this, first choose an ordered basis for the $(2t-s)$-dimensional sum of
the two row spaces, and then choose, inside this fixed space, the two
$t$-dimensional subspaces and their ordered bases.  The latter choices
contribute only a constant depending on $q$ and $t$.

Let $J_s$ denote
the total of the joint probabilities over all ordered pairs with
intersection dimension $s$.  Multiplying these two bounds and then using
\eqref{eq:dimension-slack} gives
\[
J_s\leq c'_{q,t}q^{(2t-s)k}q^{n((2t-s)h-(2t-s))}
=q^{(2t-s)a_n+O(1)}.
\]
On the other hand, the lower bound for the mean gives
\[
\mu^2\geq c_{q,t}^2(n+1)^{-2q^t}q^{2ta_n}.
\]
The polynomial factors in these estimates can be bounded independently of
$s$ for $1\leq s\leq t-1$.  Thus there is a constant $K=K(q,t)$ such that
\begin{equation}\label{eq:normalized-intersection-contribution}
\frac{J_s}{\mu^2}\leq q^{-sa_n+K\log_q(n+1)}
\leq (n+1)^{-sL+K}.
\end{equation}
Increasing $L$ if necessary so that $L\geq K+2$, the right-hand side is
$O((n+1)^{-2})$ simultaneously for every $1\leq s\leq t-1$.

Combining the preceding bounds and using
$\Exp[I_AI_B]-p_n^2\leq\Exp[I_AI_B]$, we obtain
\begin{equation}\label{eq:variance-assembly}
\operatorname{Var}(X_{\mathbf v})\leq O(\mu)+\sum_{s=1}^{t-1}J_s.
\end{equation}
Here the $O(\mu)$ term comprises the diagonal sum and the pairs with $s=t$;
the pairs with $s=0$ have zero covariance, and the remaining pairs are
controlled after normalization by
\eqref{eq:normalized-intersection-contribution}.  Since $t$ is fixed,
dividing \eqref{eq:variance-assembly} by $\mu^2$ and using
\eqref{eq:inverse-mean} yields
\[
\frac{\operatorname{Var}(X_{\mathbf v})}{\mu^2}=O((n+1)^{-2}).
\]
Finally, Chebyshev's inequality yields
\[
\Prob[X_{\mathbf v}=0]\leq
\frac{\operatorname{Var}(X_{\mathbf v})}{\mu^2}=O((n+1)^{-2}),
\]
with constants independent of $\mathbf v$.
\end{proof}

For a realization of $G$, define the target matrices not covered by the
selected type class as
\begin{equation}\label{eq:uncovered-set}
\mathcal U(C):=\F_q^{t\times n}\setminus(C^t+\mathcal E_n).
\end{equation}
Since $X_{\mathbf v}>0$ implies
$\mathbf v\in C^t+\mathcal E_n$, membership in $\mathcal U(C)$ implies
$X_{\mathbf v}=0$.  Hence Proposition~\ref{prop:almost-covering} and
linearity of expectation give
\[
\begin{aligned}
\Exp_G\left[\frac{|\mathcal U(C)|}{q^{tn}}\right]
&=\frac1{q^{tn}}\sum_{\mathbf v\in\F_q^{t\times n}}
  \Prob_G[\mathbf v\in\mathcal U(C)]
\\
&\leq\frac1{q^{tn}}\sum_{\mathbf v\in\F_q^{t\times n}}
  \Prob_G[X_{\mathbf v}=0]\\
&=O((n+1)^{-2}).
\end{aligned}
\]
There is therefore a realization of $G$ for which
\begin{equation}\label{eq:small-uncovered-density}
\delta_n:=\frac{|\mathcal U(C)|}{q^{tn}}=O((n+1)^{-2}).
\end{equation}
The random linear code therefore covers all but a vanishing fraction of
the target matrices with the desired dimension bound.  The next subsection
converts this into complete coverage.

\subsection{Structured Alteration and Proof of the Linear Theorem}

\subsubsection{Squaring the Uncovered Density}

Keep the error family $\mathcal E_n$ fixed, and define $\mathcal U(C)$ by
\eqref{eq:uncovered-set} for every linear code $C\leq\F_q^n$.  The
following lemma isolates the density-squaring mechanism of the structured
alteration based on \cite[Proposition~21]{elimelech-firer-schwartz}.
Rather than adding one center for each remaining target matrix, the
alteration adds at most $t$ dimensions and squares the entire uncovered
density.

\begin{lemma}\label{lem:density-squaring}
Let $C\leq\F_q^n$ and
$\delta:=|\mathcal U(C)|/q^{tn}$.  There is a matrix
$Y\in\F_q^{t\times n}$ such that the linear code
$C':=C+\Row(Y)$
satisfies
\[
\dim_{\F_q}C'\leq\dim_{\F_q}C+t,
\qquad \frac{|\mathcal U(C')|}{q^{tn}}\leq\delta^2.
\]
\end{lemma}

\begin{proof}
For every $Y\in\F_q^{t\times n}$, the set $(C')^t$ contains both $C^t$
and $C^t+Y$.  Consequently,
\[
(C')^t+\mathcal E_n
\supseteq(C^t+\mathcal E_n)\cup\bigl((C^t+\mathcal E_n)+Y\bigr).
\]
Taking complements in $\F_q^{t\times n}$ and using
$\F_q^{t\times n}\setminus((C^t+\mathcal E_n)+Y)=\mathcal U(C)+Y$ gives
\[
\mathcal U(C')\subseteq\mathcal U(C)\cap\bigl(\mathcal U(C)+Y\bigr).
\]
Averaging the intersection on the right over all $Y$ gives
\[
\frac1{q^{tn}}\sum_{Y\in\F_q^{t\times n}}
|\mathcal U(C)\cap(\mathcal U(C)+Y)|
=\frac{|\mathcal U(C)|^2}{q^{tn}}.
\]
Indeed, every ordered pair of elements of $\mathcal U(C)$ determines
exactly one translation $Y$.  Some $Y$ therefore makes the intersection
in the displayed inclusion at most
$|\mathcal U(C)|^2/q^{tn}$.  After division by $q^{tn}$, this is
$\delta^2$.
\end{proof}

\subsubsection{Completion of the Linear Case}

We first present the quantitative conclusion of the random construction
and alteration.

\begin{theorem}\label{thm:finite-length}
Let $q$ be a prime power, let $t\geq2$ be fixed, and let
$0<\theta<1-q^{-t}$ be rational.  For every sufficiently large length $n$,
there is a linear code
$C_n\leq\F_q^n$ satisfying
\[
R_t(C_n)\leq\theta n,
\qquad
\dim_{\F_q}C_n\leq n\bigl(1-H_{q^t}(\theta)\bigr)+O(\log n).
\]
The implied constant depends only on $q,t$, and $\theta$.
\end{theorem}

\begin{proof}
Start with the code $C=\Row(G)$ obtained above, whose uncovered density
satisfies \eqref{eq:small-uncovered-density}.  For all sufficiently large
admissible $n$, we have $\delta_n\leq1/2$.  Apply
Lemma~\ref{lem:density-squaring} repeatedly.  After $\ell$ applications,
the uncovered density is at most
\[
\delta_n^{2^\ell}\leq2^{-2^\ell}.
\]
Choose
\[
\ell:=\left\lceil\log_2(tn\log_2q)\right\rceil+1=O(\log n).
\]
Then $q^{tn}2^{-2^\ell}<1$.  Since the number of uncovered target matrices
is an integer, no target matrix remains.  Every error in $\mathcal E_n$ has
$t$-weight $\theta n$, so the resulting code $C_n$ satisfies
$R_t(C_n)\leq\theta n$.

Each application of Lemma~\ref{lem:density-squaring} adds at most $t$
dimensions.  Using \eqref{eq:dimension-choice}, we obtain
\[
\begin{aligned}
\dim_{\F_q}C_n
&\leq k+t\ell=n(1-h)+O(\log n)\\
&=n\bigl(1-H_{q^t}(\theta)\bigr)+O(\log n).
\end{aligned}
\]
This proves the assertion at every sufficiently large admissible length.

To remove the integrality restriction, let $D$ be a positive integer such
that every multiple of $D$ is admissible.  For an arbitrary sufficiently
large length $N$, put $n=D\lfloor N/D\rfloor$ and $d=N-n<D$, and define
\[
\widetilde C_N:=C_n\times\F_q^d\leq\F_q^N.
\]
The final $d$ coordinates can match every target matrix exactly, so
$R_t(\widetilde C_N)=R_t(C_n)\leq\theta n\leq\theta N$.  Moreover,
\[
\dim_{\F_q}\widetilde C_N
\leq n(1-h)+O(\log n)+d=N(1-h)+O(\log N),
\]
because $d<D$ is bounded independently of $N$.  Relabeling $N$ as $n$
completes the proof.
\end{proof}

We now deduce Theorem~\ref{thm:linear-rate} from the finite-length
bound.

\begin{proof}[Proof of Theorem~\ref{thm:linear-rate}]
For $t=1$, the asserted formula is the classical linear-covering result of
Cohen and Frankl \cite{Cohen85}.  We therefore assume $t\geq2$.

The lower bound in the range $0\leq\rho<1-q^{-t}$ is
\eqref{eq:volume-lower-bound}.  To prove the matching linear upper bound,
first let $0<\rho<1-q^{-t}$ and choose any rational
$\theta\in(0,\rho)$.  Theorem~\ref{thm:finite-length} provides codes for
every sufficiently large block length with
\[
\frac{\dim_{\F_q}C_n}{n}
\leq1-H_{q^t}(\theta)+O\left(\frac{\log n}{n}\right),
\qquad R_t(C_n)\leq\theta n<\rho n.
\]
Hence $\kappa_t^{\mathrm{Lin}}(\rho,q)\leq1-H_{q^t}(\theta)$.  Letting
$\theta\to\rho^-$ through rational values and using continuity gives
\[
\kappa_t^{\mathrm{Lin}}(\rho,q)\leq1-H_{q^t}(\rho).
\]

For $\rho=0$, the full space $\F_q^n$ has covering radius zero and rate
one, which matches the volume bound.  If
$1-q^{-t}\leq\rho\leq1$, choose rational
$\theta<1-q^{-t}$ arbitrarily close to $1-q^{-t}$.  Then
$\theta<\rho$ and $H_{q^t}(\theta)$ tends to one, so
Theorem~\ref{thm:finite-length} gives linear codes of arbitrarily small
asymptotic rate.  Thus,
$\kappa_t^{\mathrm{Lin}}(\rho,q)=0$ throughout this range.
This completes the proof.
\end{proof}

\section{Concluding Remarks}\label{sec:conclusion}

We determined the optimal asymptotic rate of generalized covering codes
for every alphabet size and fixed order, and proved that the same rate is
attained by linear codes over every finite field.  Thus, both the
product-form constraint on the covering centers and the additional
linearity constraint are asymptotically cost-free.  This resolves both
problems posed by Elimelech and Schwartz
\cite{elimelech-schwartz}.

The two proofs use the balanced exact-type family in different ways.  Its
row-agreement bounds support the Janson argument for
unrestricted codes, while its row-difference bounds support the
second-moment argument for linear codes.

Both proofs use only the non-strict normalized marginal-entropy bound in
Lemma~\ref{lem:entropy-balance}.  In fact, the
Appendix shows that, if $X\sim P_\theta$, then
\[
\frac{H(X_S)}{|S|}>\frac{H(X)}{t}
\qquad(\emptyset\neq S\subsetneq[t]).
\]
Although unnecessary for the main results, this strict entropy gap permits
the use of the ordinary Janson inequality in the unrestricted proof and
may be useful in related problems.

A natural open problem is to give explicit constructions of generalized
covering codes attaining the optimal rate in both the unrestricted and
linear settings.

\appendix

\section{A Strict Normalized Marginal-Entropy Gap}
\label{app:strict-gap}

We prove the strict strengthening of Lemma~\ref{lem:entropy-balance} and
then explain how it gives an alternative treatment of the Janson step in
the unrestricted proof.

\begin{proposition}
Let $q\geq2$, $t\geq2$, and $0<\theta<1-q^{-t}$, and let
$X=(X_1,\ldots,X_t)$ have the distribution $P_\theta$ defined in
Section~\ref{sec:common}.  Then, for every nonempty proper
subset $S\subsetneq[t]$, we have
$H(X_S)/|S|>H(X)/t$.
\end{proposition}

\begin{proof}
For $m=1,\ldots,t$, let $X_{<m}=(X_1,\ldots,X_{m-1})$, with $X_{<1}$
empty, and set $a_m:=H(X_m\mid X_{<m})$.
Exchangeability gives
$H(X_m\mid X_{<m})=H(X_{m+1}\mid X_{<m})$, and hence, for $m<t$,
\[
a_m-a_{m+1}=I(X_m;X_{m+1}\mid X_{<m}).
\]
Here $I$ denotes conditional mutual information.

It remains to show that this conditional mutual information is positive.
Set $\alpha:=\theta/(1-q^{-t})$, so that $0<\alpha<1$.  The random
vector $X$ can be generated using a Bernoulli random variable $Z$ with
$\Prob[Z=1]=\alpha$: if $Z=0$, set $X=\overline 0$, while if $Z=1$, choose
$X_1,\ldots,X_t$ independently and uniformly from $G_q$.  This gives
exactly the distribution $P_\theta$.

Fix $m<t$ and condition on the positive-probability event
\[
E_m:=\{X_1=\cdots=X_{m-1}=0\};
\]
here $E_1$ is the sure event.  Writing
$\beta:=\Prob[Z=1\mid E_m]$, we have
$\beta=\alpha q^{-(m-1)}/(1-\alpha+\alpha q^{-(m-1)})\in(0,1)$.
Consequently,
\begin{align*}
\Prob[X_m\neq0,\ X_{m+1}\neq0\mid E_m]
&=\beta(1-q^{-1})^2,\\
\Prob[X_m\neq0\mid E_m]
\Prob[X_{m+1}\neq0\mid E_m]
&=\beta^2(1-q^{-1})^2.
\end{align*}
The two quantities are different, so $X_m$ and $X_{m+1}$ are not
conditionally independent given $E_m$.  Since $E_m$ specifies one value
of $X_{<m}$ having positive probability, it follows that
\[
I(X_m;X_{m+1}\mid X_{<m})>0.
\]
Thus $a_1>\cdots>a_t$.  If $s=|S|<t$, exchangeability gives
$H(X_S)=H(X_1,\ldots,X_s)$, while the entropy chain rule gives
\[
H(X_1,\ldots,X_s)=\sum_{m=1}^s a_m,
\qquad H(X)=\sum_{m=1}^t a_m.
\]
Therefore,
\[
\frac{H(X_S)}{s}=\frac1s\sum_{m=1}^s a_m
>\frac1t\sum_{m=1}^t a_m=\frac{H(X)}{t},
\]
as required.
\end{proof}

Although this strict inequality is not needed for either main theorem, it
gives an alternative treatment of the Janson step in
Section~\ref{sec:nonlinear-proof}.  Retain the notation
$\rho_0,\eta,\mu$, and $\Delta_{\mathbf v}$ used there, take
$\theta=\rho_0$, and set
\[
\varepsilon_\theta
:=\min_{\emptyset\neq S\subsetneq[t]}
\left(H(P_{\theta,S})-|S|H_{q^t}(\theta)\right)>0.
\]
If the sampling slack is chosen so that
$\eta<\varepsilon_\theta/(2(2t-1))$, then combining
\eqref{eq:type-class-row-difference-count} at $\mathbf z=\mathbf 0$ with
\eqref{eq:janson-dependency-general} and the definition of $p$ gives
\[
\Delta_{\mathbf v}\leq q^{n(-\varepsilon_\theta+(2t-1)\eta)+O(1)}
=q^{-\Omega(n)}.
\]
Consequently, the ordinary Janson inequality
\cite[Theorem~8.1.1]{alon-spencer} yields
\[
\Prob[X_{\mathbf v}=0]
\leq\exp(-\mu+\Delta_{\mathbf v}/2)
=\exp\left(-q^{nt\eta+o(n)}\right),
\]
which is again sufficient for the union bound over all target matrices.

\section*{Acknowledgment}

\emph{Generative-AI use disclosure:} During the preparation of this
manuscript, the authors used OpenAI Codex (GPT-5.6 Sol model, Ultra mode) as
an auxiliary tool for exploring proof ideas, checking
calculations, and improving the exposition. All mathematical content,
including statements, proofs, and references, was independently verified by
the authors, who assume full responsibility for the final manuscript.


\begin{thebibliography}{99}

\bibitem{alfarano2026geometric}
G.~N.~Alfarano, G.~Marino, A.~Neri, and R.~Trombetti,
``A geometric approach to generalized covering radii of linear codes,''
arXiv:2606.16669, 2026.

\bibitem{alon-spencer}
N.~Alon and J.~H.~Spencer,
\emph{The Probabilistic Method},
4th ed. Hoboken, NJ, USA: Wiley, 2016.

\bibitem{Astola09}
J.~T.~Astola and R.~S.~Stankovic,
``Application of covering codes for reduced representations of logic
functions,'' in \emph{Proc. 39th Int. Symp. Multiple-Valued Logic}, 2009,
pp.~304--311.

\bibitem{barouch-essayag-zabokritskiy2026}
I.~Barouch Essayag and A.~L.~Zabokritskiy (Yohananov),
``The exact second generalized covering radius of binary primitive
triple-error-correcting BCH codes,'' arXiv:2608.07215, 2026.

\bibitem{Bierbrauer2008}
J.~Bierbrauer and J.~Fridrich,
``Constructing good covering codes for applications in steganography,''
in \emph{Transactions on Data Hiding and Multimedia Security III},
Lecture Notes in Computer Science, vol.~4920. Berlin, Germany: Springer,
2008, pp.~1--22, doi: 10.1007/978-3-540-69019-1\_1.

\bibitem{chung-graham-frankl-shearer}
F.~R.~K.~Chung, R.~L.~Graham, P.~Frankl, and J.~B.~Shearer,
``Some intersection theorems for ordered sets and graphs,''
\emph{J. Combin. Theory Ser. A}, vol.~43, no.~1, pp.~23--37, 1986,
doi: 10.1016/0097-3165(86)90019-1.

\bibitem{cohen-covering-codes}
G.~Cohen, I.~Honkala, S.~Litsyn, and A.~Lobstein,
\emph{Covering Codes},
Amsterdam, The Netherlands: North-Holland, 1997.

\bibitem{Cohen85}
G.~Cohen and P.~Frankl,
``Good coverings of Hamming spaces with spheres,''
\emph{Discrete Math.}, vol.~56, nos.~2--3, pp.~125--131, 1985.

\bibitem{cover-thomas}
T.~M.~Cover and J.~A.~Thomas,
\emph{Elements of Information Theory},
2nd ed. Hoboken, NJ, USA: Wiley, 2006.

\bibitem{csiszar-korner}
I.~Csisz\'ar and J.~K\"orner,
\emph{Information Theory: Coding Theorems for Discrete Memoryless Systems},
2nd ed. Cambridge, U.K.: Cambridge Univ. Press, 2011.

\bibitem{elimelech-firer-schwartz}
D.~Elimelech, M.~Firer, and M.~Schwartz,
``The generalized covering radii of linear codes,''
\emph{IEEE Trans. Inf. Theory}, vol.~67, no.~12, pp.~8070--8085, 2021.

\bibitem{elimelech-schwartz}
D.~Elimelech and M.~Schwartz,
``The second-order football-pool problem and the optimal rate of
generalized-covering codes,'' \emph{J. Combin. Theory Ser. A}, vol.~203,
Art. no.~105834, 2024.

\bibitem{elimelech-wei-schwartz}
D.~Elimelech, H.~Wei, and M.~Schwartz,
``On the generalized covering radii of Reed--Muller codes,''
\emph{IEEE Trans. Inf. Theory}, vol.~68, no.~7, pp.~4378--4391, 2022.

\bibitem{Fernandes1983}
H.~Fernandes and E.~Rechtschaffen,
``The football pool problem for 7 and 8 matches,''
\emph{J. Combin. Theory Ser. A}, vol.~35, pp.~109--114, 1983.

\bibitem{Hamalainen1995}
H.~H{\"a}m{\"a}l{\"a}inen, I.~Honkala, S.~Litsyn, and
P.~R.~J.~{\"O}sterg{\aa}rd,
``Football pools---a game for mathematicians,''
\emph{Amer. Math. Monthly}, vol.~102, no.~7, pp.~579--588, 1995.

\bibitem{Kamps1967}
H.~J.~Kamps and J.~H.~van Lint,
``The football pool problem for 5 matches,''
\emph{J. Combin. Theory Ser. A}, vol.~3, pp.~315--325, 1967.

\bibitem{Linderoth09}
J.~Linderoth, F.~Margot, and G.~Thain,
``Improving bounds on the football pool problem by integer programming and
high-throughput computing,'' \emph{INFORMS J. Comput.}, vol.~21, no.~3,
pp.~445--457, 2009.

\bibitem{Lovasz1975}
L.~Lov\'asz,
``On the ratio of optimal integral and fractional covers,''
\emph{Discrete Math.}, vol.~13, no.~4, pp.~383--390, 1975.

\bibitem{luo-et-al2026}
R.~Luo, Z.~Zhou, S.~Mesnager, V.~Sagar, and H.~Yan,
``Determining the exact value of the second-order generalized covering
radius of two classes of binary cyclic codes,''
\emph{Des. Codes Cryptogr.}, vol.~94, Art. no.~71, 2026,
doi: 10.1007/s10623-025-01799-2.

\bibitem{Micciancio2004}
D.~Micciancio,
``Almost perfect lattices, the covering radius problem, and applications to
Ajtai's connection factor,'' \emph{SIAM J. Comput.}, vol.~34, no.~1,
pp.~118--169, 2004.

\bibitem{Ostergard1996}
P.~R.~J.~{\"O}sterg{\aa}rd,
``A combinatorial proof for the football pool problem for six matches,''
\emph{J. Combin. Theory Ser. A}, vol.~76, pp.~160--163, 1996.

\bibitem{ozbudak-ozturk-second2026}
F.~{\"O}zbudak and \.{I}.~{\"O}zt{\"u}rk,
``On the second generalized covering radius for binary primitive
triple-error-correcting BCH codes,'' \emph{Des. Codes Cryptogr.}, vol.~94,
Art. no.~165, 2026, doi: 10.1007/s10623-026-01903-0.

\bibitem{ozbudak-ozturk2026}
F.~{\"O}zbudak and \.{I}.~{\"O}zt{\"u}rk,
``The third generalized covering radius for binary primitive
double-error-correcting BCH codes,'' \emph{Finite Fields Appl.}, vol.~110,
Art. no.~102749, 2026, doi: 10.1016/j.ffa.2025.102749.

\bibitem{Stein1974}
S.~K.~Stein,
``Two combinatorial covering theorems,''
\emph{J. Combin. Theory Ser. A}, vol.~16, no.~3, pp.~391--397, 1974.

\bibitem{VKWei1991}
V.~K.~Wei,
``Generalized Hamming weights for linear codes,''
\emph{IEEE Trans. Inf. Theory}, vol.~37, no.~5, pp.~1412--1418, 1991.

\bibitem{xiong-yip2026}
M.~Xiong and C.~H.~Yip,
``On generalized covering radii of binary primitive double-error-correcting
BCH codes,'' arXiv:2603.21068, 2026.

\bibitem{yohananov-schwartz2025}
L.~Yohananov and M.~Schwartz,
``The second generalized covering radius of binary primitive
double-error-correcting BCH codes,'' \emph{Finite Fields Appl.}, vol.~107,
Art. no.~102648, 2025, doi: 10.1016/j.ffa.2025.102648.

\bibitem{fulman-goldstein2015}
J.~Fulman and L.~Goldstein, ``Stein's method and the rank distribution of random matrices over finite fields,''
\emph{Ann. Probab.},
vol. 43, no. 3, pp.~1274--1314, 2015.
\end{thebibliography}
\end{document}